\documentclass[journal]{IEEEtran}

\usepackage{etex}
\normalsize
\usepackage{hyperref}  
\usepackage[T1]{fontenc}
\usepackage{amsthm,amsmath,amssymb,amsfonts, cuted}
\usepackage{mathrsfs}
\usepackage{mathabx}
\usepackage{amsbsy}
\usepackage{graphicx}
\usepackage{graphics}
 \usepackage{epstopdf}
\usepackage{algorithm}
\usepackage{subfigure}
\usepackage{cite}
\usepackage{multirow}
\usepackage{ctable}
\usepackage{enumerate}
\usepackage{caption}
\usepackage{setspace}
\usepackage{mathtools}
\usepackage{lipsum}
\usepackage[algo2e]{algorithm2e}
\usepackage{diagbox}
\usepackage[thinc]{esdiff}
\usepackage{algpseudocode}

\newcommand{\name}{\textsf{Info-Commit}}
\newcommand{\namespace}{\name{ }}
\DeclareMathOperator{\rank}{rank}
\DeclareMathAlphabet{\mathbfsl}{OT1}{cmr}{b}{it} 
\DeclareMathAlphabet\mathbfcal{OMS}{cmsy}{b}{n}
\newcommand*\bell{\ensuremath{\boldsymbol\ell}}


\title{\name: Information-Theoretic \\ Polynomial Commitment}

\author{%
	 Saeid Sahraei, Salman Avestimehr and Ramy E. Ali
	\vspace{-.25in}
	\thanks{Saeid Sahraei was with the Department of Electrical Engineering, University of Southern California, Los Angeles, CA 90089 USA. He is now with the Department of Wireless Research and
Development, Qualcomm Technologies, Inc., San Diego, CA 92121 USA
(e-mail: \href{mailto:ssahraei@qti.qualcomm.com}{ssahraei@qti.qualcomm.com}).}
	\thanks{Salman Avestimehr and Ramy E. Ali are with the Department of Electrical Engineering, University of Southern California, Los Angeles, CA 90089 USA (e-mail:\href{mailto:avestimehr@ee.usc.edu}{avestimehr@ee.usc.edu} and \href{mailto:reali@usc.edu}{reali@usc.edu}).}
}

\newtheorem{example}{Example}
\newtheorem{theorem}{{Theorem}}
\newtheorem{lemma}{Lemma}

\newtheorem{definition}{{Definition}}

\AfterEndEnvironment{remark}{\noindent\ignorespaces}

\begin{document}

\maketitle

\begin{abstract}
We introduce \name, an information-theoretic protocol for polynomial commitment and verification. \textcolor{black}{With the help of a trusted initializer, a succinct commitment to a private polynomial $f$ is provided to the user}. The user then queries the server to obtain evaluations of $f$ at several inputs \textcolor{black}{chosen by the user}. The server provides the evaluations along with proofs of correctness which the user can verify against the initial commitment. \namespace has four main features. Firstly, the user is able to detect, with high probability, if the server has responded with evaluations of the same polynomial initially committed to. Secondly, \namespace provides rigorous privacy guarantees for the server: upon observing the initial commitment and the response provided by the server to $m$ evaluation queries, the user only learns $O(m^2)$ symbols about the coefficients of $f$. Thirdly, the verifiability and the privacy guarantees are unconditional regardless of the computational power of the two parties. Lastly, \namespace is doubly-efficient in the sense that in the evaluation phase, the user runs in $O(\sqrt{d})$ time and the server runs in $ O(d)$ time, where $d-1$ is the degree of the polynomial $f$.
\end{abstract}

\section{Introduction}
Consider a user who wishes to delegate the task of running a proprietary software, such as an advanced optimization or a machine learning algorithm, on his data to an untrusted server. To enhance the credibility of the results, the server sends a proof of correctness along with the outcome of the computation to the user. \textcolor{black}{On the other hand, the user may not be authorized to access the underlying software. Therefore, while the proof must be informative enough to convince the user of the correctness of the result, it should reveal as little information as possible about the underlying program}. To enable this process, one could envision a trusted party who provides the user with a verification key. This party could be, for instance, the developer of the software. While the server may be incentivized to satisfy the customers with as little computational resources as possible, the developer has little reason to provide the users with a false verification key. 

The literature on verifiable computing \cite{goldwasser2008delegating,gennaro2010non,benabbas2011verifiable,fiore2012publicly} and functional commitment protocols \cite{libert2016functional,kate2010constant} allow the server (prover) to generate a proof of correctness of \textcolor{black}{the results which can be efficiently checked by the users (verifiers). Such an algorithm is said to have the {\it soundness} property, if a malicious prover cannot provide a convincing proof of an incorrect evaluation to the verifier. In general, the soundness of these algorithms relies on heuristic hardness assumptions. Our goal in this paper is to design a commitment and verification algorithm that overcomes this limitation.}  Namely, (i) the verifiability property must hold regardless of the computation power of the prover, and (ii) the verifier must not learn more than a constant number of bits about the underlying function, regardless of \textcolor{black}{the verifier's} computational power. In the example above, the prover should be able to keep \textcolor{black}{the} program private from the verifier while at the same time providing him with verifiable results. This verification is performed against an initial commitment by the prover which contains almost no information about the program itself. Once the verifier receives the commitment to a specific function, the prover will only be able to pass the verification process \textcolor{black}{by providing evaluations of the same function that the prover has committed to}. 
Our focus in this work is on functions that can be represented as polynomials. This encompasses a wide range of applications including transaction verification in cryptocurrencies \cite{polyshard, kadhe2019sef, rana2020free2shard, cao2020cover}, verifiable secret sharing \cite{chor1985verifiable,kate2010constant}, and proof-of-storage \cite{zheng2012secure,benet2017proof,fiore2012publicly}.

\subsection{Our contributions}
We propose \name, an information-theoretic protocol for polynomial commitment and verification. \namespace consists of a commitment phase and an evaluation phase. In the commitment phase, a third-party initializer enables the verifier to learn a private commitment to the prover's private polynomial $f$ of degree $d-1$ over a finite field $\mathbb F_q$. 
In the evaluation phase, the verifier requests evaluations of the polynomial $f$ at many input points from the prover. \namespace is verifiable, privacy-preserving and doubly-efficient. More specifically, \namespace has the following salient features.
\begin{enumerate}
    \item {\bf Verifiability.} If the prover commits to a polynomial $f$ in the commitment phase, \textcolor{black}{the prover} must provide evaluations of the same polynomial $f$ in the evaluation phase. Otherwise, the verifier can detect the inconsistency with high probability. On the other hand, if the prover is honest, the verifier accepts the results with probability $1$. 
    
    Moreover, \namespace is adaptive as it remains secure even when the prover is informed about the verification outcome of the previous queries. 
\item {\bf Privacy.} The other facet of \namespace is its privacy-preserving property. Upon observing the initial commitment and the response of the prover to $m$ rounds of evaluation, the verifier only learns $O(m^2)$ symbols over the finite field about the coefficients of the polynomial $f$.

\item {\bf Information-theoretic guarantees.} Both the verifiability and the privacy guarantees of \namespace are information-theoretic, meaning that the verifier and the prover can have unbounded computational power and may arbitrarily deviate from the prescribed protocol. 
\item {\bf Double-efficiency.} For each   evaluation round, the verifier runs in $O(\sqrt{d})$ time, whereas the prover runs in $O(d)$ time. Note that we measure the efficiency of \namespace in the amortized model \cite{gennaro2010non}, where the initial commitment phase is followed by many rounds of evaluations and the one-time cost of the commitment phase can be neglected.
\end{enumerate}

\subsection{Related work}

The closest works in the literature to this work are commitment protocols, including polynomial commitment \cite{kate2010constant,papamanthou2013signatures,ma2013verifiable,bultel2017verifiable}, vector commitment \cite{catalano2013vector}, cryptographic accumulators \cite{benaloh1993one}, and more generally functional commitments \cite{libert2016functional}. These protocols allow a prover to provide a cryptographic commitment to a specific function which reveals little to no information about the underlying function. A verifier can then request the prover to ``open" the commitment at specific locations. The verifier should be able to detect if the opening is inconsistent with the initial commitment. \\
There are several factors that distinguish the model in the current paper from the traditional notion of functional commitment. The {\it first} difference is in the notion of privacy. The linear commitment model in \cite{libert2016functional} requires the initial commitment to reveal no information about the function. However, it does not impose any requirement on how much information is revealed upon observing the evaluations (opening) of the functions along with the initial commitment. On the other hand, the polynomial commitment scheme in \cite{kate2010constant} only requires the {\it evaluations} of the polynomial at unqueried points to remain hidden from the verifier, which is a rather weaker privacy constraint. In contrast, we impose rigorous {\it information-theoretic} constraints on how much information is revealed about the {\it coefficients} of the polynomial upon observing the initial commitment {\it and} the evaluations at a constant number of input points. {\it Secondly}, while the literature on functional commitment is more concerned with the size of the commitment and the size of the opening, we instead focus on designing doubly-efficient algorithms, i.e., algorithms with efficient provers and super-efficient verifiers. Furthermore, unlike the literature on functional commitment, our verifiability guarantee is information-theoretic.

Another relevant concept is zero-knowledge verifiable computation \cite{parno2013pinocchio} and zero-knowledge arguments of knowledge \cite{groth2016size,ben2019aurora,wu2018dizk,sasson2014zerocash}. Such algorithms enable a verifier and a prover to interact to compute $f(v,u)$, where $f$ is a publicly known function, $v$ is the input of the verifier and $u$ is the private input of the prover. The verifier will not learn anything about $u$ except for what is implied through $f(v,u)$. Furthermore, the verifier will be convinced that there exists some $u$ known to the prover such that the computed value corresponds to $f(v,u)$. 
In the context of polynomial evaluation, $v$ could be the input to the polynomial, $u$ the coefficients of the polynomial, and $f$ an operator that maps $(v,u)$ to the evaluation of the polynomial. Unfortunately, this approach does not provide any binding guarantees, as the prover may change \textcolor{black}{the} polynomial for every input.

Other related notions in the literature are oblivious polynomial evaluation \cite{naor2006oblivious,hazay2018oblivious,tassa2013oblivious,tonicelli2015information} which guarantees the privacy of both parties in a semi-honest setting, and verifiable computation \cite{gennaro2010non,fiore2012publicly,sahraei2019interpol,sahraei2021interactive} which guarantees efficient verifiability, but generally ignores the privacy of the prover. \name, however, jointly ensures privacy of the prover and efficient verifiability.
\textcolor{black}{\subsection{Relation with INTERPOL \cite{sahraei2019interpol}}
INTERPOL was presented in \cite{sahraei2019interpol} as an information-theoretic alternative to the existing cryptographic solutions for verifiable polynomial evaluation. In that context, both the verifier and the prover had access to the underlying polynomial. \name, on the other hand, is focused on the notion of polynomial commitment. The new challenge is to keep the underlying polynomial hidden from the verifier, while convincing the verifier that the results conform with an initial commitment. The evolution from INTERPOL to the present work has been outlined in Section \ref{sec:overview}. Section \ref{subsec:basic} can be seen as a variation of INTERPOL which relies on a trusted initializer to eliminate the need for the verifier to access the entire polynomial. Section \ref{subsec: semi-honest verifier} enables privacy against a semi-honest verifier. Finally, Section \ref{subsec: malicious verifier} illustrates how to achieve privacy against a malicious verifier. }

 \textbf{Organization.} The rest of this paper is organized as follows. We provide a brief background and formally present the problem statement  in Section \ref{sec:statement}. An overview of \namespace is provided in Section \ref{sec:overview} and then Section \ref{sec:formal} describes \namespace in detail.
In Section \ref{sec:analysis}, we analyze the correctness, the complexity, the soundness and the privacy guarantees of \name. Finally, in Section \ref{sec:discussion}, we discuss some concluding remarks.

\textcolor{black}{
\section{Preliminaries and Problem Statement}
\label{sec:statement}
In this section, we provide a brief background and describe the polynomial commitment and verification problem.
\subsection{Preliminaries}
\label{subsec:preliminaries}
The following definition and property will be useful in the analysis of \name.
\begin{definition}[Markov Chain] An ordered set of three random variables $(X,Y, Z)\in {\cal X}\times {\cal Y}\times {\cal Z}$ is said to form a Markov chain if 
\begin{align}
    \mathbb{P}(X = x,Z=z|Y = y) = \mathbb{P}(X= x|Y=y)\mathbb{P}(Z=z|Y=y)
\end{align}
for every $(x,y,z)\in {\cal X}\times {\cal Y}\times {\cal Z}$. Intuitively, this means that the two random variables $X$ and $Z$ are independent, conditioned on knowing $Y$.
This Markov chain is represented by $X\longleftrightarrow Y \longleftrightarrow Z$.
\end{definition}
\textbf{Property 1.} If $X\longleftrightarrow Y \longleftrightarrow Z$ forms a Markov chain, then $H_q(X|Y,Z) = H_q(X|Y)$, where $H_q$ represents the Shannon entropy in base $q$.}

\subsection{Polynomial \textcolor{black}{c}ommitment and \textcolor{black}{v}erification}
\label{subsec:statement}
Suppose a prover is in possession of a private polynomial $f(x) = a_0 + a_1x + \cdots + a_{d-1}x^{d-1}$ selected uniformly at random among all polynomials of degree $d-1$ over $\mathbb{F}_q$. A verifier, who only knows the degree of the polynomial, wishes to obtain evaluations of this polynomial at several input points $x\in\mathbb{F}_q$ and verify the correctness of the results in sub-linear time in $d$. Before this evaluation phase, the prover commits to the polynomial $f$ during an initialization phase. Once this commitment is done, the prover is expected to evaluate the same polynomial $f$ for every input $x \in \mathbb F_q$ provided by the verifier. The verifier should be able to detect, efficiently and with high probability, if for any input $x\in\mathbb{F}_q$ the prover returns an incorrect evaluation $y \neq f(x)$. We assume that the verifier is interested in evaluating the polynomial $f$ at many input points, so that the complexity of the commitment phase amortizes over many rounds of the evaluation phase \cite{gennaro2010non}. Therefore, efficiency is only measured with respect to the evaluation phase. Additionally, during the entire commitment and evaluation phases, the verifier should \textcolor{black}{only learn a negligible amount of information about the coefficients of the polynomial $f$}. Note that revealing a constant number of symbols over $\mathbb{F}_q$ about the coefficients of $f$ to the verifier is inevitable, since such amounts of information can be learned by investigating a single pair $(x,f(x))$. Formally, a polynomial commitment and verification protocol consists of the following algorithms.

\begin{itemize}
    \item {\bf Commit$(a_0,\cdots,a_{d-1},K_v,K_p)$}. In the commitment phase, the verifier and the prover engage in a commitment protocol \textbf{Commit}$(a_0,\cdots,a_{d-1},K_v,K_p)$ to help the verifier learn a secret verification key VK which depends on \textcolor{black}{the verifier's} secret key, $K_v$, the prover's secret key, $K_p$, and the polynomial coefficients $a_0,\cdots,a_{d-1}$.
    \item {\bf Eval$(x,a_0,\cdots,a_{d-1},K_p)$}. The verifier requests an evaluation of $f$ at $x$ from the prover. The prover then provides $val = ${\bf Eval}$(x,a_0,\cdots,a_{d-1},K_p)$ to help the verifier evaluate $f(x)$.
    \item {\bf Verify}$(x,val,\text{VK},K_v)$. The verifier checks the correctness of $val$ based on \textcolor{black}{$K_v$ and VK}. If $val$ is incompatible with the initial commitment, \textcolor{black}{the verifier} rejects the evaluation. Otherwise, \textcolor{black}{the verifier} proceeds to recover $f(x)$.
    \item {\bf Recovery$(x,val)$}. If the verification process passes, the verifier will proceed to recover the evaluation $f(x)$ via an algorithm {\bf Recovery}$(x,val)$.
\end{itemize}
We are interested in polynomial commitment and verification protocols that are information-theoretically private and efficiently verifiable as defined next.
\begin{definition}[Information-theoretically Private and Verifiable Polynomial Commitment]
A polynomial commitment and verification protocol is information-theoretically private and verifiable if it satisfies the following properties.
\begin{itemize}
    \item {\bf Correctness}. If the prover follows the {\bf Commit} and the {\bf Eval} algorithms for the same polynomial $f$, then the {\bf Verify} algorithm must return $1$ with probability $1$.
    \item {\bf (Information-theoretic) Soundness}. If the prover commits to a polynomial with coefficients $a_0,\cdots,a_{d-1}$, then the probability that \textcolor{black}{the prover} can pass the verification test with $\hat{val}\neq \text{\bf Eval}(x,a_0,\cdots,a_{d-1},K_p)$ should be negligible, regardless of \textcolor{black}{the prover's} computational power. That is, we must have
    \begin{align}
        \Pr(&\text{\bf Verify}(x,\hat{val},\text{VK},K_v,K_p) = 1, \notag \\ &\hat{val} \neq \text{\bf Eval}(x,a_0,\cdots,a_{d-1},K_p)) = o(1),
    \end{align}
    where the term $o(1)$ vanishes as the finite field size $q$ grows. 
    \item {\bf Efficient Verification}. The two functions {\bf Verify} and {\bf Recovery} must run in sub-linear time in $d$.
    \item {\bf Efficient Evaluation}. The running time of {\bf Eval} must be comparable to the time required for evaluating the polynomial $f$, i.e., linear in $d$.
    \item {\bf (Information-theoretic) Privacy}. After running {\bf Commit} and {\bf Eval} for $m$ different inputs $x_1,\cdots,x_m$, the verifier should only learn $\rho = O(\mathrm{poly}(m))$ symbols over the field about the coefficients of the polynomial $f$, regardless of \textcolor{black}{the verifier's} computation power. Importantly, $\rho$ should be independent of the degree of the polynomial. That is, we must have
    \begin{align}
        &H_q(a_0,\cdots,a_{d-1}|\text{VK}, (x_1,val_1),\cdots,(x_m,val_m) ) \notag \\ &= d - O(\mathrm{poly}(m)).
    \end{align}
\end{itemize}
\label{def:requirements}
\end{definition}

\section{An Overview of \name}
\label{sec:overview}
In this section, we provide an overview of \namespace. We first note that the polynomial $f(x) = a_0 + a_1 x + \cdots + a_{d-1}x^{d-1}$ can be expressed as follows 
\begin{align}
    f(x) = \begin{bmatrix} 1 & x^s & \cdots &x^{s(s-1)}\end{bmatrix}
\mathbfsl A
    \begin{bmatrix}1 & x & \cdots & x^{s-1}\end{bmatrix}^{\mathrm T},
\end{align}
where $s = \sqrt{d}$, and 
\begin{align}
\mathbfsl A =     \begin{bmatrix}a_{0} & a_1 & \cdots & a_{s-1}\\
                    a_s & a_{s+1} & \cdots & a_{2s-1}\\
                    &\ddots &&\\
                    a_{s^2-s} & a_{s^2-s+1} & \cdots & a_{s^2-1} 
    \end{bmatrix}.
    \label{eqn:poly_to_matrix}
\end{align}
\subsection{A basic algorithm}
\label{subsec:basic}
A straightforward algorithm is for the verifier to generate a $c\times s$ matrix $\mathbfsl \Lambda\in \mathbb{F}_q^{c\times s}$ uniformly at random, for some security parameter $c$. In the commitment phase, the verifier and the prover will engage in a commitment protocol through the trusted initializer to help the verifier learn 
\begin{align}
\label{eqn:simple_commitment}
\mathbfsl \Gamma = \mathbfsl \Lambda \mathbfsl A.
\end{align}
In this process, the verifier will not learn anything about $\mathbfsl A$ other than what is revealed through $\mathbfsl \Gamma$, and the prover will not learn anything about $\mathbfsl \Lambda$. \textcolor{black}{For instance, the verifier and the prover respectively provide $\mathbfsl \Lambda$ and $\mathbfsl A$ to the trusted party, and the trusted party sends $\mathbfsl \Gamma$ to the verifier}. In the evaluation phase, the verifier will ask the prover to compute 
\begin{align}
\mathbfsl b = \mathbfsl A \begin{bmatrix}1 & x & \cdots & x^{s-1}\end{bmatrix}^{\mathrm T}.
\end{align}
Suppose the prover responds with $\hat{\mathbfsl b}$. The verifier checks if
\begin{align}
\mathbfsl \Lambda \hat{\mathbfsl b} = \mathbfsl \Gamma  \begin{bmatrix}1 & x & \cdots & x^{s-1}\end{bmatrix}^{\mathrm T},    
\end{align}
which can be done in $O(\sqrt{d})$ time. If $\hat{\mathbfsl b}\neq \mathbfsl b$, this identity holds only with probability $q^{-c}$ \cite{sahraei2019interpol}, since the prover does not know anything about $\mathbfsl \Lambda$.  If the verification process is successful, the verifier can recover 
\begin{align}
\hat{f}(x) = \begin{bmatrix} 1 & x^s & \cdots &x^{s(s-1)}\end{bmatrix} \hat{\mathbfsl b}
\end{align}
in $O(\sqrt{d})$ time and accepts it as the correct evaluation. \textcolor{black}{This algorithm has been previously explored in \cite{sahraei2019interpol}.}\\ The main issue with this algorithm is that $\mathbfsl \Gamma$ reveals $O(\sqrt{d})$ symbols over $\mathbb{F}_q$ about the polynomial $f$ during the commitment phase, and another $\sqrt{d}$ symbols for each evaluation $\mathbfsl b$. We wish to reduce this to a constant. We will accomplish this in two steps: firstly, against a semi-honest verifier in Section \ref{subsec: semi-honest verifier} and next against a malicious verifier in Section \ref{subsec: malicious verifier}.

\subsection{Improved privacy against a semi-honest verifier}
\label{subsec: semi-honest verifier}
In order to improve the privacy of this basic algorithm, the prover will generate a polynomial $g$ of degree $d-1$ uniformly among all polynomials of degree $d-1$ over $\mathbb{F}_q$. In the commitment phase, instead of directly committing to $f$, the prover commits to both $h  = f+ g$ and $g$. Let $\mathbfsl A$ and $\mathbfsl B$ represent the $s\times s$ matrices corresponding to the polynomials $f$ and $g$, constructed similarly to \eqref{eqn:poly_to_matrix}. If both commitments have the same structure as in \eqref{eqn:simple_commitment}, the verifier will be able to gain substantial information about the matrix $\mathbfsl A$ by choosing \textcolor{black}{the} two secret matrices to be linearly dependent. 

To prevent this, the protocol will require the prover to commit to $\mathbfsl A+ \mathbfsl B$ ``from the left" as in \eqref{eqn:simple_commitment}, whereas the commitment to $\mathbfsl B$ will be done ``from the right". More concretely, the verifier and the prover engage in commitment algorithm through the trusted initializer such that the verifier learns 
\begin{align}
\label{eqn:advanced_commitment}
    \mathbfsl \Gamma &= \mathbfsl \Lambda(\mathbfsl A+\mathbfsl B),\nonumber\\
    \mathbfsl \Omega &= \mathbfsl B \mathbfsl \Theta^{\mathrm T},
\end{align}
for two $c\times s$ secret matrices $\mathbfsl \Lambda$ and $\mathbfsl \Theta$ generated \textcolor{black}{uniformly at random} by the verifier. To gain intuition on how this \textcolor{black}{preserves} the privacy of $\mathbfsl A$, it helps to consider the case of $c=1$. Knowing any single linear combination of the rows of $\mathbfsl A+ \mathbfsl B$ and any single linear combination of the columns of $\mathbfsl B$ cannot reveal more than one symbol over $\mathbb{F}_q$ about the elements of the matrix $\mathbfsl A$. This argument will be made precise in the privacy analysis of  \namespace in Section \ref{sec:analysis}. In the evaluation phase, the prover is \textcolor{black}{asked} to compute  
\begin{align}
\mathbfsl v &= (\mathbfsl A+ \mathbfsl B) \begin{bmatrix} 1 & x & \cdots & x^{s-1}\end{bmatrix}^{\mathrm T},\nonumber\\
\mathbfsl u &= \begin{bmatrix} 1 & x^s & \cdots & x^{s(s-1)}\end{bmatrix} \mathbfsl B.
\label{eqn:vector_uv}
\end{align}
Suppose the prover responds with $\hat{\mathbfsl v}$ and $\hat{\mathbfsl u}$. The verifier will check the correctness of each result via two parities 
\begin{align}
\mathbfsl \Gamma \begin{bmatrix} 1 & x & \cdots & x^{s-1}\end{bmatrix}^{\mathrm T} &= \mathbfsl \Lambda \hat{\mathbfsl v},\nonumber\\
    \begin{bmatrix} 1 & x^s & \cdots & x^{s(s-1)}\end{bmatrix} \mathbfsl \Omega &=  \hat{\mathbfsl u} \mathbfsl \Theta^{\mathrm T}.  
    \label{eqn:verify_uv}
\end{align}
If both verification tests are successful, then the verifier recovers 
\begin{align}
\hat{h}(x) &= \begin{bmatrix} 1 & x^s & \cdots & x^{s(s-1)}\end{bmatrix} \hat{\mathbfsl v},\nonumber\\
\hat{g}(x) &= \hat{\mathbfsl u}\begin{bmatrix} 1 & x & \cdots & x^{s-1}\end{bmatrix}^{\mathrm T}.
\label{eqn:recover}
\end{align}
Finally, the verifier will find 
\begin{align}
\hat{f}(x) = \hat{h}(x) - \hat{g}(x)
\end{align}
and will accept $\hat{f}(x)$ as the correct evaluation of $f$ at $x$. This algorithm guarantees the verifiability of the results with high probability following a similar analysis to \cite{sahraei2019interpol}. Nonetheless, the privacy only holds as long as $(\mathbfsl \Lambda, \mathbfsl \Theta)$ are generated \textcolor{black}{uniformly at random over $\mathbb{F}_q^{c\times s}$}, i.e., if the verifier is semi-honest. We will overcome this limitation in the next subsection.

\subsection{Improving privacy against a malicious verifier}
\label{subsec: malicious verifier}
If both parties follow the protocol of Section \ref{subsec: semi-honest verifier} as described, the verifier will only learn a constant number of symbols about the polynomial $f$ upon observing $(\mathbfsl \Gamma,\mathbfsl \Omega)$ and the pair $(\mathbfsl v, \mathbfsl u)$ for a constant number of inputs (see the analysis of privacy in Section \ref{sec:analysis}). However, a malicious verifier may deviate from this protocol by choosing $\mathbfsl \Lambda$ and $\mathbfsl \Theta$ in a deterministic manner as in the following example. 
\begin{example}[Malicious Verifier]~
\label{ex: malicious}
\begin{enumerate}
\item \textbf{Commitment Phase}. Suppose that $c =1$ and the verifier chooses  $\mathbfsl \Lambda = \begin{bmatrix}1 & 0 & \cdots & 0\end{bmatrix}$. This helps \textcolor{black}{the verifier} learn $\mathbfsl \Gamma = \mathbfsl A_1 + \mathbfsl B_1$, the first row of the matrix $\mathbfsl A+ \mathbfsl B$. 
\item \textbf{Evaluation Phase}. In this phase, the verifier requests the $(\mathbfsl v, \mathbfsl u)$ pair for $x=0$. Based on this, \textcolor{black}{the verifier} can learn $\mathbfsl u = \mathbfsl B_1$, which is the first row of the matrix $\mathbfsl B$. \item The verifier then \textcolor{black}{computes} $\mathbfsl \Gamma - \mathbfsl u$ in order to find $\mathbfsl A_1$, which reveals $\sqrt{d}$ symbols about the polynomial $f$.
\end{enumerate}
\end{example}
We will propose a mechanism to resolve this threat. Suppose the verifier requests the $(\mathbfsl v, \mathbfsl u)$ pair for $m$ different inputs $x_1,\cdots,x_m$. Define the matrices $\mathbfsl X$ and $\mathbfsl Y$ as follows 
\begin{align}
    \mathbfsl X &= \begin{bmatrix}
        1 & x_1 & \cdots & x_1^{s-1}\\
        & \cdots &\cdots &\\
        1 & x_m & \cdots & x_m^{s-1}
    \end{bmatrix},\nonumber\\
    \mathbfsl Y&= \begin{bmatrix}
        1 & x_1^s & \cdots & x_1^{s(s-1)}\\
        & \cdots &\cdots &\\
        1 & x_m^s & \cdots & x_m^{s(s-1)}
    \end{bmatrix},
    \label{eqn:bigXY}
\end{align}
and let the matrices $\mathbfsl U$ and $\mathbfsl V$ be the concatenation of the vectors $\mathbfsl u$ and $\mathbfsl v$ for all $x\in\{x_1,\cdots,x_m\}$ as follows 
\begin{align}
    \mathbfsl V &= (\mathbfsl A+ \mathbfsl B) \mathbfsl X^{\mathrm T},\nonumber\\
    \mathbfsl U &= \mathbfsl Y \mathbfsl B.
    \label{eqn:bigUV}
\end{align}
Since after $m$ rounds of evaluation, the verifier learns  $\mathbfsl \Gamma = \mathbfsl \Lambda(\mathbfsl A+ \mathbfsl B)$ and $\mathbfsl U= \mathbfsl Y \mathbfsl B$, the prover must ensure that the matrices $\mathbfsl \Lambda$ and $\mathbfsl Y$ do not contain any linear dependencies among their rows. Otherwise, this linear dependency can be exploited to extract substantial information about the matrix $\mathbfsl A$ as in Example \ref{ex: malicious}. Since $\rank(\mathbfsl Y) = m$ and $\rank(\mathbfsl \Lambda) = c$ (with high probability), the prover must ensure that 
\begin{align}
\rank(\mathbfsl Y|| \mathbfsl \Lambda) = c + m,
\end{align}
where $||$ denotes the vertical concatenation of the two matrices. Similarly, we must have 
\begin{align}
\rank(\mathbfsl X||\mathbfsl \Theta) = c+m.
\end{align}
To accomplish this, we restrict the matrices $\mathbfsl \Lambda$ and $\mathbfsl \Theta$ to have the same Vandermonde structures as the matrices $\mathbfsl Y$ and $\mathbfsl X$, respectively. In other words, $\mathbfsl \Lambda$ and $\mathbfsl \Theta$ must be of the form 
\begin{align}
\mathbfsl \Lambda &=   
\begin{bmatrix}
    1 & \lambda^s_1 & \cdots & \lambda_1^{s(s-1)}\\
     & \cdots & \cdots& \\
    1 & \lambda^s_c & \cdots & \lambda_c^{s(s-1)}
    \end{bmatrix},\nonumber\\
    \mathbfsl \Theta &= \begin{bmatrix}
    1 & \theta_1 & \cdots & \theta_1^{s-1}\\
     & \cdots & \cdots& \\
    1 & \theta_c & \cdots & \theta_c^{s-1}
    \end{bmatrix}.
    \label{eqn:restricted}
\end{align}
Note that as long as 
\begin{align}
|\{\lambda_1,\cdots,\lambda_c,x_1,\cdots,x_m\}| &= |\{\theta_1,\cdots,\theta_c,x_1,\cdots,x_m\}| \notag \\
&= c + m,
\label{eqn:full_xy}
\end{align}
the two rank requirements will be satisfied. Note also that, without loss of generality, we can assume that $\lambda_i \neq \lambda_j$, $\theta_i \neq \theta_j$ and $x_i \neq x_j$ for $i\neq j$, since the verifier cannot benefit from receiving the same evaluation twice. 
Restricting the structure of $\mathbfsl \Lambda$ and $\mathbfsl \Theta$ as opposed to choosing them uniformly at random over $\mathbb{F}_q^{c\times s}$ would mean that the prover now has some side information about these matrices which \textcolor{black}{can be utilized to bypass the verification test}. Fortunately, the analysis in Section \ref{sec:analysis} shows that this probability remains sufficiently small.

It is left to convince the prover that \eqref{eqn:full_xy} holds. For this purpose, we designate a set $S\subseteq \mathbb{F}_q$ of prohibited values from which the elements $\lambda_i$ and $\theta_i$ can be chosen. ``{Prohibited}" means that in the evaluation phase, the verifier is not permitted to delegate the evaluation of the function $f$ at any member of $S$ to the prover. For instance, the verifier could provide an upper-bound $\xi$ on its input values to the prover, and then $S\subseteq \mathbb{F}_q$ could be chosen as a sufficiently large interval whose smallest member is strictly larger than $\xi$.

\subsection{Considerations regarding the size of the field}
The privacy analysis of \namespace in Section \ref{sec:analysis} relies on the fact that the set $\{\alpha^s| \alpha \in S \}$ is of the same size as $S$, where $S$ is the set of prohibited values. This property holds if the function $g(x) = x^s$ is a permutation over $\mathbb{F}_q$ which is the case if $\mathrm{gcd}(s,q-1) =1$. It is also important that there exists a prohibited subset of $\mathbb{F}_q$ of sufficiently large size. To ensure this, either the verifier should set \textcolor{black}{the} upper bound $\xi$ appropriately, or the field size should be increased once the verifier provides the upper-bound. Of course, this adjustment must be done without violating the first property.

\section{Formal Description of \name}
\label{sec:formal}

The formal description of the four functions of \name, namely {\bf Commit}, {\bf Eval}, {\bf Verify} and {\bf Recovery} are provided in Algorithms \ref{Alg:commit},\ref{Alg:eval},\ref{Alg:verify} and \ref{Alg:recovery}, respectively. Theorem \ref{thm:thm} establishes that \namespace satisfies the requirements of an information-theoretically verifiable and private protocol for commitment and verification.
 \begin{algorithm}[H]
\caption[caption]{The commitment function,\\ {\bf Commit}$(a_0,\cdots,a_{d-1},K_v,K_p)$}
\begin{algorithmic}[1]
\Statex {{\bf Input:} The private polynomial coefficients, $a_0,\cdots,a_{d-1}$, the verifier's secret key, $K_v$, the prover's secret key, $K_p$.}
\Statex {\bf Output:} The verifier learns the verification key VK$=(\mathbfsl \Gamma, \mathbfsl \Omega)$.
\Statex{}
    \State The verifier sends an upper-bound $\xi$ on its maximum input to the prover. The two parties agree on a prohibited set $S\subset \mathbb{F}_q$ of size $r(s-1)$, where $s=\sqrt{d}$ and $r$ is a small positive integer (for instance, $r=10)$. Every element of $S$ must be greater than $\xi$.
    \State The prover generates a random polynomial $g$ of degree $d-1$ with a matrix representation $\mathbfsl B$. Let $h = f + g$ and $K_p = \mathbfsl B$.
    \State The verifier chooses $c$ distinct elements of $S$ uniformly at random, denoted as $\lambda_1,\cdots,\lambda_c$, where $c$ is a small positive integer (for instance, $c=10$). Similarly, he chooses $\theta_1,\cdots,\theta_c$. Let $K_v = (\lambda_1,\cdots,\lambda_c,\theta_1,\cdots,\theta_c)$ and define $\mathbfsl \Lambda$ and $\mathbfsl \Theta$ as in \eqref{eqn:restricted}. 
    \State The verifier sends $K_v$ to the trusted initializer. The prover also sends $K_p$ and $a_0, \cdots, a_{d-1}$ to the trusted initializer.
   \State The trusted initializer sends the verification key VK$ = (\mathbfsl \Gamma,\mathbfsl \Omega)$ to the verifier, where $\mathbfsl \Gamma = \mathbfsl \Lambda( \mathbfsl A+ \mathbfsl B)$ and $\mathbfsl \Omega = \mathbfsl B \mathbfsl \Theta^{\mathrm T}$.
\end{algorithmic}
\label{Alg:commit}
\end{algorithm}
\begin{theorem}
\label{theorem:InfoCommit}
\namespace as described by the {\bf Commit}, {\bf Eval}, {\bf Verify} and {\bf Recovery} Algorithms \ref{Alg:commit},\ref{Alg:eval},\ref{Alg:verify} and \ref{Alg:recovery} satisfies the correctness, soundness, privacy and efficiency requirements stated in Definition \ref{def:requirements} as follows. 
\label{thm:thm}
\begin{itemize}
     \item If the prover is honest, the verifier will accept the results with probability $1$.
     \item The probability that the prover can pass the verification process with an incorrect result is negligible, 
    \begin{align}
        \Pr(\text{\bf Verify}(x,\hat{val},\text{VK}, K_v) = 1,\hat{val} \neq {val}) \le \frac{2}{r^{c}}+\frac{1}{r^{2c}},
    \end{align}
    where $r=|S|/(\sqrt{d}-1)$, $S \subset \mathbb F_q$ is the prohibited  set and $x \notin S$.
    \item The verifier only learns $O(m^2)$ symbols about the coefficients of the polynomial $f$ after \textcolor{black}{the verifier} receives $(\text{VK},val_1,\cdots,val_m)$ for any choice of $K_v$ and any $m$ input values $x_1,\cdots,x_m$,
        \begin{align}
       & H_q(a_0,\cdots,a_{d-1}|\text{VK}, (x_1,val_1),\cdots,(x_m,val_m)) \notag \\ &= d - O(m^2).
    \end{align}
    \item The complexity of the verifier per evaluation round is $O(\sqrt{d})$,
    \begin{align}
        {\cal C}_{\text{\bf Verify}} + {\cal C}_{\text{\bf Recovery}} = O(\sqrt{d}),
    \end{align}
where  ${\cal C}_{\text{\bf Verify}}$ and  ${\cal C}_{\text{\bf Recovery}}$ denote the complexity of the {\bf Verify} and the {\bf Recovery} algorithms, respectively.   
    \item The complexity of the prover per evaluation round is $O(d)$,
    \begin{align}
        {\cal C}_{\text{\bf Eval}} = O(d),
    \end{align}
where ${\cal C}_{\text{\bf Eval}}$ denotes the complexity of the {\bf Eval} algorithm.     
\end{itemize}
\end{theorem}

 \begin{algorithm}[H]
\caption[caption]{The evaluation function,\\ {\bf Eval}$(x,a_0,\cdots,a_{d-1},K_p)$}
\begin{algorithmic}[1]
\Statex {{\bf Input:} The coefficients of the polynomial, $a_0,\cdots,a_{d-1}$, the prover's secret key, $K_p = \mathbfsl B$ and the input to the polynomial, $x \notin S$.}
\Statex {{\bf Output:} $val = (\mathbfsl v, \mathbfsl u)$.}

        \Statex The prover evaluates 
\begin{align}
\mathbfsl v &= (\mathbfsl A+ \mathbfsl B) \begin{bmatrix} 1 & x & \cdots & x^{s-1}\end{bmatrix}^{\mathrm T},\nonumber \\
\mathbfsl u &= \begin{bmatrix} 1 & x^s & \cdots & x^{s(s-1)}\end{bmatrix} \mathbfsl B,
\label{eqn:uv_final}
\end{align}
and returns $val =(\mathbfsl v, \mathbfsl u)$ to the verifier.
\end{algorithmic}
\label{Alg:eval}
\end{algorithm}
 
 \begin{algorithm}[H]
\caption[caption]{The verification function, \\ {\bf Verify}$(x,\hat{val},\text{VK},K_v)$}
\begin{algorithmic}[1]
\Statex {{\bf Input:} The input to the polynomial, $x$, the evaluation provided by the prover, $\hat{val} = (\hat{\mathbfsl v},\hat{\mathbfsl u})$, the (secret) verification key, VK$=(\mathbfsl \Gamma, \mathbfsl \Omega)$, and the verifier's secret key, $K_v$.}
\Statex {{\bf Output:} The verifier either accepts or rejects $\hat{val}$.}
    \Statex   The verifier checks the two parities 
    \begin{align}
\mathbfsl \Gamma \begin{bmatrix} 1 & x & \cdots & x^{s-1}\end{bmatrix}^{\mathrm T} &= \mathbfsl \Lambda \hat{\mathbfsl v},\nonumber\\
    \begin{bmatrix} 1 & x^s & \cdots & x^{s(s-1)}\end{bmatrix} \mathbfsl \Omega &=  \hat{\mathbfsl u} \mathbfsl \Theta^{\mathrm T}.  
    \label{eqn:verification_final}
\end{align}
If either parity fails, then {\bf Verify} returns $0$ (rejected), otherwise, it returns $1$.
\end{algorithmic}
\label{Alg:verify}
\end{algorithm}

 \begin{algorithm}[H]
\caption[caption]{The recovery function, {\bf Recovery}$(x,\hat{val})$}
\begin{algorithmic}[1]
\Statex {{\bf Input:} The input to the polynomial, $x$, and the evaluation provided by the prover, $\hat{val} = (\hat{\mathbfsl v},\hat{\mathbfsl u})$.}
\Statex {{\bf Output:}  $\hat{f}(x)$ based on $\hat{val}$.}
        \Statex If {\bf Verify} returns $1$, the verifier will compute $\hat{h}(x)$ and $\hat{g}(x)$ based on 
\begin{align}
\hat{h}(x) &= \begin{bmatrix} 1 & x^s & \cdots & x^{s(s-1)}\end{bmatrix} \hat{\mathbfsl v}\nonumber\\
\hat{g}(x) &= \hat{\mathbfsl u}\begin{bmatrix} 1 & x & \cdots & x^{s-1}\end{bmatrix}^\mathrm{T}
\end{align}
and finds $\hat{f}(x) = \hat{h}(x) - \hat{g}(x)$
\end{algorithmic}
\label{Alg:recovery}
\end{algorithm}

\textcolor{black}{We now illustrate the performance of \name. Fig. \ref{fig:privacy} compares the amount of information revealed to the verifier after $m$ queries against the entropy of the secret polynomial. Fig.  \ref{fig:soundness} illustrates the probability of success of a malicious prover as a function of the parameter $c$, assuming $r=10$ where $r$ is defined in Algorithm \ref{Alg:commit}. Finally, Fig. \ref{fig:complexity} compares the complexity of the verifier relying on \name\space against the complexity of computing $f(x)$ without help. Note that the complexity is only evaluated orderwise, and can be viewed as the number of elementary computations up to a constant multiplicative factor.}

\begin{figure*}[htb!]
\centering
    \subfigure[\footnotesize Privacy of \name.]{\label{fig:privacy}
    \includegraphics[scale=0.35]{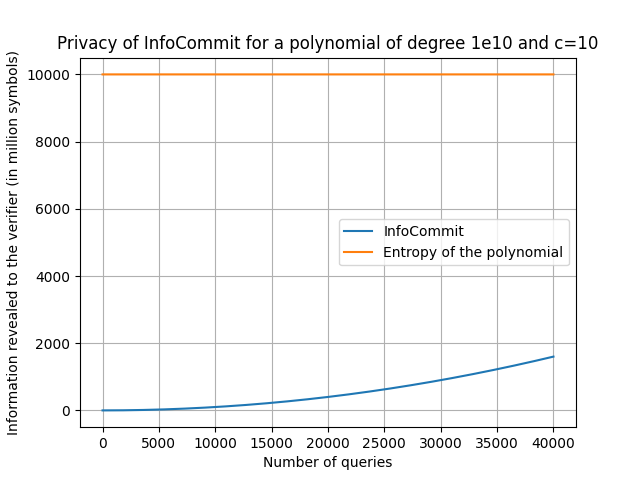}
    }
    \subfigure[\footnotesize Soundness of \name.]{\label{fig:soundness}
    \includegraphics[scale=0.35]{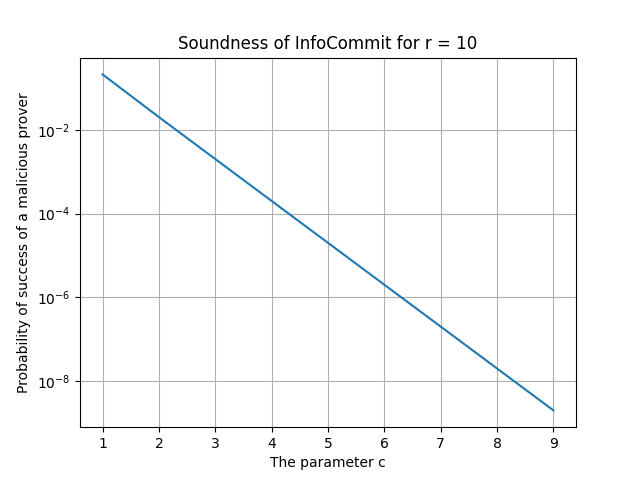}
    }
     \subfigure[\footnotesize Complexity of verifier relying on \name.]{\label{fig:complexity}
    \includegraphics[scale=0.35]{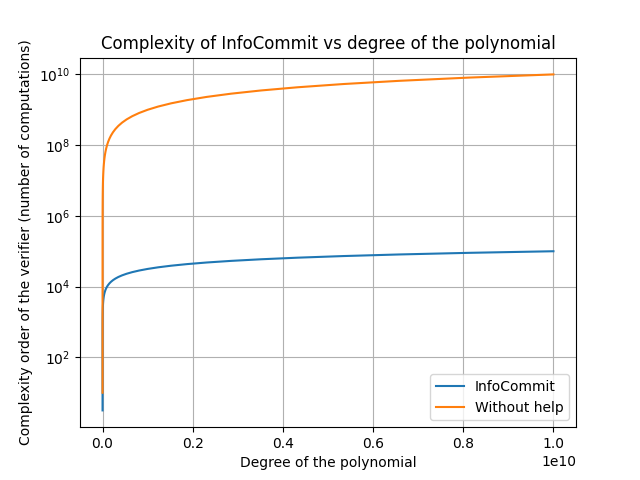}
    }
\vspace{-0.2cm}
\caption{An illustration of the privacy, soundness and complexity of \name.}
\label{fig:key_metrics}
\end{figure*}

\textbf{Adaptivity of \name}. Theorem \ref{theorem:InfoCommit} considers the case where the user does not inform the server whether he accepted or rejected the results of the previous queries. We next study the  effect of such feedback on the soundness of \name. Specifically, we consider a setting with $m$ sequential queries $\{x_1, x_2, \cdots, x_m\}$ such that the user reveals the verification outcome to the server after each round. In such a setting, we show that the soundness guarantee of \namespace changes only slightly as given in Theorem \ref{theorem:Adaptivity InfoCommit}. \textcolor{black}{The proof closely follows the proof of adaptivity of INTERPOL as presented in \cite{sahraei2019interpol}.}
\begin{theorem}[Adaptivity of \name]
\label{theorem:Adaptivity InfoCommit}
In a setting with $m$ sequential queries $\{x_1, x_2, \cdots, x_m\}$ such that the verification outcome is revealed to the prover after each round, we have 
    \begin{align}
        \Pr(\exists &\ell \in[m] \ \text{s.t.} \ \text{\bf Verify}(x_\ell,\hat{val}_\ell,\text{VK}, K_v) = 1,\notag \\ &\hat{val}_\ell \neq {val}_\ell) \le m \left(\frac{2}{r^{c}}+\frac{1}{r^{2c}}\right),
    \end{align}
where $r=|S|/(\sqrt{d}-1)$, $S \subset \mathbb F_q$ is the prohibited  set and $x_\ell \notin S, \ell \in [m]$.
\end{theorem}
\noindent The proof of Theorem \ref{theorem:Adaptivity InfoCommit} is provided in Appendix \ref{appendix: Proof of Adaptivity thm}.

\section{The Analysis of \namespace \\(Proof of Theorem \ref{theorem:InfoCommit})} 
\label{sec:analysis}
In this section, we prove Theorem \ref{theorem:InfoCommit}.

\noindent{\bf Correctness}. If the prover is honest, he will provide the evaluations of the same polynomial $f$ that he initially committed to. Consequently, in the verification phase,  \eqref{eqn:verification_final}  holds and the recovered value is equal to $h(x) - g(x) = f(x)$ as desired.\\
\noindent{\bf Efficiency}. The verification can be done in $O(\sqrt{d})$ since it only requires multiplying $c\times s$ matrices by vectors of length $s = \sqrt{d}$. The recovery can be done in $O(\sqrt{d})$ since the main operation in this phase is inner products of vectors of length $s$. Therefore, the overall complexity of the verifier in the evaluation phase is $O(\sqrt{d})$. The prover, on the other hand, must compute the two vectors $\mathbfsl v$ and $\mathbfsl u$ according to \eqref{eqn:uv_final} which can be done in  $O(d)$.\\
\textcolor{black}{\noindent{\bf Soundness}. The  soundness analysis follows a similar logic as the proof of soundness of INTERPOL \cite{sahraei2019interpol}. Specifically, the analysis relies on the fact that the prover learns nothing about the secret values $\lambda_i,\theta_i$, $i\in[c]$ during the commitment phase. Note however that unlike \name, INTERPOL imposed no structural limitation on the secret matrix ${ \mathbfsl \Lambda}$. We must show that despite this structural limitation, soundness is guaranteed with overwhelming probability.} \\ Without loss of generality, we can assume that the prover has committed to the correct polynomial $f$. If not, we simply use the letter $f$ to denote the polynomial that the prover has committed to and expect the soundness property to hold with respect to this $f$. Note that such a polynomial of degree $d-1$ exists regardless of how the prover responds to the queries in the commitment phase. Remember that for each input $x$, the prover is required to provide the verifier with two vectors $\mathbfsl v$ and $\mathbfsl u$ as defined in \eqref{eqn:uv_final}. Here, we analyze the probability that a malicious prover can pass the verification test \eqref{eqn:verification_final} with a $(\hat{\mathbfsl v}, \hat{\mathbfsl u}) \neq (\mathbfsl v, \mathbfsl u)$ and show that this probability is negligible. Specifically, a malicious prover may respond with (1) $\hat{\mathbfsl v} \neq \mathbfsl v$ and $\hat{\mathbfsl u}=\mathbfsl u$, (2) $\hat{\mathbfsl v}=\mathbfsl v$ and  $\hat{\mathbfsl u} \neq \mathbfsl u$ or (3) $\hat{\mathbfsl v} \neq \mathbfsl v$ and $\hat{\mathbfsl u} \neq \mathbfsl u$. 
Let $({\cal V},{\cal U})$ denote a random variable from which the prover draws the evaluation $(\hat{\mathbfsl v},\hat{\mathbfsl u})$. Note that $({\cal V},{\cal U})$ is independent of $({ \mathbfsl \Lambda},{\mathbfsl \Theta})$. We next define these events 
\begin{align}
&E_1=\{\mathcal V \neq \mathbfsl v: \mathbfsl \Lambda (\mathcal V-\mathbfsl v) = \mathbfsl 0 \}, \\
&E_2=\{ \mathcal U \neq \mathbfsl u: (\mathcal U-\mathbfsl u) \mathbfsl \Theta^{\mathrm T} = \mathbfsl 0\} \ \text{and} \\
&E_{12}=\{ \mathcal V \neq \mathbfsl v, \mathcal U \neq \mathbfsl u: \mathbfsl \Lambda (\mathcal V-\mathbfsl v) = \mathbfsl 0, (\mathcal U-\mathbfsl u) \mathbfsl \Theta^{\mathrm T} = \mathbfsl 0\}.
\end{align}
The probability that a malicious prover can pass the verification then can be upper-bounded as follows 
\begin{align}
&\Pr({\text{\bf Verify}(x,\hat{val},\text{VK},K_v) = 1, \hat{val}\neq val}) \notag \\ 
&\leq \Pr(E_1)+\Pr(E_2)+\Pr(E_{12}).
\end{align}
We first consider the probability of the first event $E_1$ as follows
\begin{align}
\label{eqn:first-parity-verification}
    p_{\mathbfsl v}  \coloneqq \Pr(E_1) = \Pr\Big(\mathbfsl \Lambda ({\cal V}-\mathbfsl v) = {\bf 0} , {\cal V} \neq \mathbfsl v\Big),
\end{align}
where the randomness is with respect to $(\mathbfsl \Lambda,{\cal V})$. We need an upper bound on $p_{\mathbfsl v}$ that holds for any $x$ and any $\mathbfsl A, \mathbfsl B$. Define $\begin{bmatrix}\gamma_0,\cdots,\gamma_{s-1}\end{bmatrix} = {\cal V}-\mathbfsl v$ which can be viewed as the coefficients of a polynomial of degree $s-1$, $\gamma(x) = \gamma_0 + \gamma_1 x + \cdots + \gamma_{s-1}x^{s-1}$. Equation (\ref{eqn:first-parity-verification}) then represents the probability that the prover can provide a nontrivial polynomial of degree $s-1$ such that all $\lambda_i^s$, $i\in [c]$, are the roots of this polynomial. Note that the polynomial $\gamma(x)$ has at most $s-1$ roots in the set $S' = \{y^s|y\in S\}$. Note also that all the terms $\lambda_i^s, i\in[c]$, are distinct. This is because the field size $q$ is chosen in such a way that $\mathrm{gcd}(s,q-1) =1$. As a result, the function $e(x) = x^s$ is a permutation, and the set $S'$ has $r(s-1)$ distinct elements. Let $T\subseteq S'$ be a set of size at most $|T| \le s-1$ that represents the roots of the polynomial $\gamma(x)$ that are in $S'$. In other words, $t\in T$, if and only if $\gamma(t) = 0$ and $t\in S'$. We are interested in the probability that $\lambda_i^s\in T$, $\forall i\in [c]$. Since ${\cal V}$ is independent of $\mathbfsl \Lambda$, so is the set $T$. We can therefore upper-bound $p_{\mathbfsl v}$ as follows 
\begin{align}
    p_{\mathbfsl v} = \Pr(\lambda_i^s \in T, \forall i\in[c])\le \frac{{s-1 \choose c}}{{r(s-1) \choose c}}\le \left(\frac{s-1}{r(s-1)} \right)^c = \frac{1}{r^c}.
\end{align}
Similarly, one can bound the probability of the second event $E_2$ as follows
 \begin{align}
 p_{\mathbfsl u} \coloneqq \Pr(E_2)=\Pr \left( (\cal U-\mathbfsl u) \mathbfsl \Theta^{\mathrm T}=\mathbfsl 0, {\cal U} \neq \mathbfsl u \right)  \le \frac{1}{r^c}.
 \end{align}
 Finally, the probability of $E_{12}$ can be also upper-bounded as  
\begin{align}
p_{\mathbfsl v, \mathbfsl u} \coloneqq \Pr(E_{12})   \leq \frac{1}{r^{2c}}.  
\end{align}
As a result, we have
 \begin{align}
 \Pr({\text{\bf Verify}(x,\hat{val},\text{VK},K_v) = 1,\hat{val}\neq val})  \leq \frac{2}{r^{c}}+\frac{1}{r^{2c}}.
 \end{align}
For instance, by choosing $r=10$ and $c=10$, this probability can be reduced to $\approx 2 \times 10^{-10}$. 
\\

We now study the privacy of \name. We begin with the following useful lemmas that will help us establish the privacy of \name.
\begin{lemma}
Let $c,d,s$ be three positive integers such that $c,d\le s$. Let $\mathbfsl E$ be an $s\times s$ random matrix uniformly distributed over $\mathbb{F}_q^{s\times s}$. Let $\mathbfsl F$ and $\mathbfsl G$ be arbitrary full-rank $c\times s$ and $s\times d$ matrices respectively. Then
\begin{align}
    H_q(\mathbfsl F \mathbfsl E | \mathbfsl E \mathbfsl G) \ge cs - cd.
\end{align}
\label{lem:conditional_matrix_entropy}
\end{lemma}
\begin{proof}
Since $\mathbfsl G$ is a full-rank matrix, then the matrix $\mathbfsl E \mathbfsl G$ is uniformly distributed over $\mathbb{F}_q^{s\times d}$ and as a result, $H_q(\mathbfsl E \mathbfsl G) = sd$. So, we only need to prove that 
\begin{align}
H_q(\mathbfsl F \mathbfsl E,\mathbfsl E \mathbfsl G) \ge cs + sd - cd.
\end{align}
Let $\bell  \in\mathbb{F}_q^{s^2}$ be a column vector obtained from the vertical concatenation of the columns of $\mathbfsl E$, such that $\ell_{i+ sj}  = E_{i,j}$ for all $i\in[0:s-1]$, $j\in[0:s-1]$. Let $\mathbfcal F  = \mathbf{I}_{s\times s}\otimes \mathbfsl F$ and $\mathbfcal G = \mathbfsl G \otimes \mathbf{I}_{s\times s}$, where $\mathbf{I}_{s\times s}$ is the $s\times s$ identity matrix and $\otimes$ denotes the Kronecker product. Since the two vectors $\mathbfcal F \bell$ and $\mathbfcal G \bell$ are rearrangements of the the two matrices $\mathbfsl  F \mathbfsl E$ and $\mathbfsl  E \mathbfsl G$, respectively, we have 
\begin{align}
H_q(\mathbfsl  F \mathbfsl  E, \mathbfsl  E \mathbfsl  G) = H_q(\mathbfcal F \bell,\mathbfcal G \bell).
\end{align}
The proof follows from the fact that the $(c+d)s \times s^2$ matrix $\mathbfcal H = \mathbfcal F||\mathbfcal G$ obtained from vertical concatenation of $\mathbfcal F$ and $\mathbfcal G$ has rank at least $(c+d)s - cd$ as long as $\mathbfsl  F$ and $\mathbfsl  G$ are full-rank (See Lemma \ref{lem:fullrank} in Appendix \ref{appendix:lemma:fullrank}).
\end{proof}
\begin{lemma}
Let $E,F,G$ be three random variables with alphabets ${\cal E}, {\cal F}, {\cal G}$, satisfying the Markov chain $E \longleftrightarrow F \longleftrightarrow G$. Let $h: {\cal E}\times{\cal G}\rightarrow {\cal W}$ be an arbitrary function. Then, we have
\begin{align}
    H(E|F, h(E,G))\ge \min_{g\in {\cal G}} H(E|F,h(E,g)).
\end{align}
\label{lem:replace_by_constant}
\end{lemma}
\begin{proof}
The following chain of inequalities proves the claim
\begin{align}
   &H(E|F, h(E,G))\ge H(E|F,G, h(E,G)) \\
   &= \sum_{g \in {\cal G}}\Pr(G = g)H(E|F,G=g, h(E,g))\\
   &\ge \min_{g\in {\cal G}}H(E|F,G=g, h(E,g))  = \min_{g\in {\cal G}}H(E|F,h(E,g)),
\end{align}
where the last inequality follows from \textcolor{black}{Property 1 in the preliminaries}. Specifically, since  $E \longleftrightarrow F \longleftrightarrow G$ forms a Markov chain, then  $E \longleftrightarrow (F,h(E,g)) \longleftrightarrow G$ forms a Markov chain too, which allows us to drop $G=g$ from conditioning. 
\end{proof}

 \textbf{Privacy}. We now prove that after $m$ evaluations, the verifier learns at most $(m+ c)^2$ symbols over $\mathbb{F}_q$ about the matrix $\mathbfsl A$ which represents the polynomial $f$.  The analysis of privacy relies on the fact that the verifier learns nothing about $(\mathbfsl A+ \mathbfsl B, \mathbfsl B)$ in the commitment phase except for what is implied via $(\mathbfsl \Gamma,\mathbfsl \Omega)$. The analysis also makes use of the fact that the $\lambda_i$ and $\theta_i$ values in the commitment phase are chosen from a ``prohibited" set. As a result of this, the matrices $\mathbfsl \Lambda || \mathbfsl Y$ and $\mathbfsl \Theta|| \mathbfsl X$ will be full-rank. Note that the verifier cannot benefit from selecting $x_i = x_j$, i.e., requesting the same evaluation point twice. Similarly, he cannot benefit from setting $\lambda_i = \lambda_j$ or $\theta_i= \theta_j$. Therefore, without loss of generality, we can assume that $\rank(\mathbfsl \Lambda ||\mathbfsl Y) = \rank(\mathbfsl \Theta||\mathbfsl X) = m + c$. \\ For $(\mathbfsl X, \mathbfsl Y, \mathbfsl V, \mathbfsl U)$ defined as in \eqref{eqn:bigXY},\eqref{eqn:bigUV}, we want to show that
\begin{align}
    H_q(\mathbfsl A| \mathbfsl V, \mathbfsl U,\mathbfsl \Gamma,\mathbfsl \Omega) \ge d - (m + c)^2,
    \label{eqn:mutual_bound}
\end{align}
for any choice of the evaluation points $x_1,\cdots,x_m$ and for any choice of $\theta_1,\cdots,\theta_c,$ $\lambda_1,\cdots,\lambda_c$. We start by simplifying the left-hand side of \eqref{eqn:mutual_bound} as follows
\small
\begin{align}
    &H_q(\mathbfsl A| \mathbfsl V, \mathbfsl U,\mathbfsl \Gamma,\mathbfsl \Omega) = H_q(\mathbfsl A| (\mathbfsl A+\mathbfsl B) \mathbfsl X^{\mathrm T},\mathbfsl Y \mathbfsl B,\mathbfsl \Lambda(\mathbfsl A+ \mathbfsl B),\mathbfsl B \mathbfsl \Theta^{\mathrm T})\nonumber\\
        &= H_q(\mathbfsl A,(\mathbfsl A+ \mathbfsl B) \mathbfsl X^{\mathrm T},\mathbfsl \Lambda (\mathbfsl A+ \mathbfsl B)| (\mathbfsl A+\mathbfsl B) \mathbfsl X^{\mathrm T},\mathbfsl Y \mathbfsl B,\mathbfsl \Lambda(\mathbfsl A+ \mathbfsl B),\mathbfsl B \mathbfsl \Theta^{\mathrm T}) \nonumber\\
    &= H_q(\mathbfsl A,\mathbfsl B \mathbfsl X^{\mathrm T},\mathbfsl \Lambda \mathbfsl B| (\mathbfsl A+ \mathbfsl B) \mathbfsl X^{\mathrm T},\mathbfsl Y \mathbfsl B,\mathbfsl \Lambda(\mathbfsl A+ \mathbfsl B),\mathbfsl B \mathbfsl \Theta^{\mathrm T}) \nonumber\\
        &= H_q(\mathbfsl B \mathbfsl X^{\mathrm T}, \mathbfsl \Lambda \mathbfsl B| (\mathbfsl A+\mathbfsl B)\mathbfsl X^{\mathrm T},\mathbfsl Y \mathbfsl B,\mathbfsl \Lambda(\mathbfsl A+ \mathbfsl B),\mathbfsl B \mathbfsl \Theta^{\mathrm T})\nonumber\\
        &+ H_q(\mathbfsl A|\mathbfsl B \mathbfsl X^{\mathrm T}, \mathbfsl \Lambda \mathbfsl B,(\mathbfsl A+\mathbfsl B) \mathbfsl X^{\mathrm T},\mathbfsl Y \mathbfsl B,\mathbfsl \Lambda(\mathbfsl A+ \mathbfsl B),\mathbfsl B \mathbfsl \Theta^{\mathrm T})\nonumber\\
    &= H_q(\mathbfsl B \mathbfsl X^{\mathrm T}, \mathbfsl \Lambda \mathbfsl B| (\mathbfsl A+ \mathbfsl B)\mathbfsl X^{\mathrm T},\mathbfsl Y \mathbfsl B,\mathbfsl \Lambda(\mathbfsl A+ \mathbfsl B),\mathbfsl B \mathbfsl \Theta^{\mathrm T})+ \notag \\ & H_q(\mathbfsl A|\mathbfsl A \mathbfsl X^{\mathrm T},\mathbfsl \Lambda \mathbfsl A),
\end{align}
\normalsize
where the last equality follows from the fact that $\mathbfsl A$ is independent of $\mathbfsl B$.
The second term in the final expression can be easily bounded as
\begin{align}
    H_q(\mathbfsl A|\mathbfsl A \mathbfsl X^{\mathrm T},\mathbfsl \Lambda \mathbfsl A) &= H_q(\mathbfsl A) - H_q(\mathbfsl A \mathbfsl X^{\mathrm T},\mathbfsl \Lambda \mathbfsl A)\nonumber\\
    &\ge H_q(\mathbfsl A) - H_q(\mathbfsl A \mathbfsl X^{\mathrm T}) - H_q(\mathbfsl \Lambda \mathbfsl A) \nonumber \\ 
    &= d - ms - cs. 
\end{align}
Therefore, we must show that the first term satisfies 
\begin{align}
&H_q(\mathbfsl B \mathbfsl X^\mathrm{T}, \mathbfsl \Lambda \mathbfsl B| (\mathbfsl A+\mathbfsl B)\mathbfsl X^\mathrm{T},\mathbfsl Y \mathbfsl B,\mathbfsl \Lambda(\mathbfsl A+\mathbfsl B),\mathbfsl B \mathbfsl \Theta^\mathrm{T}) \notag \\ &\ge ms + cs- (m + c)^2.
\end{align} 
Observe that $\Big((\mathbfsl A+\mathbfsl B)\mathbfsl X^\mathrm{T},\mathbfsl \Lambda(\mathbfsl A+\mathbfsl B)\Big)$
is independent of $\Big(\mathbfsl B \mathbfsl X^\mathrm{T},\mathbfsl \Lambda \mathbfsl B, \mathbfsl Y \mathbfsl B, \mathbfsl B \mathbfsl \Theta^\mathrm{T}\Big)$. This follows as $\mathbfsl A+\mathbfsl B$ is independent of $\mathbfsl B$. Hence, \textcolor{black}{following basic properties of the entropy function, we have} 
\begin{align}
&H_q(\mathbfsl B \mathbfsl X^\mathrm{T}, \mathbfsl \Lambda \mathbfsl B| (\mathbfsl A+ \mathbfsl B) \mathbfsl X^\mathrm{T},\mathbfsl Y \mathbfsl B,\mathbfsl \Lambda(\mathbfsl A+ \mathbfsl B),\mathbfsl B \mathbfsl \Theta^\mathrm{T}) \notag \\&= H_q(\mathbfsl B \mathbfsl X^\mathrm{T}, \mathbfsl \Lambda \mathbfsl B|\mathbfsl Y \mathbfsl B,\mathbfsl B \mathbfsl \Theta^\mathrm{T}).
\end{align}
It remains to prove that 
\begin{align}
H_q(\mathbfsl B \mathbfsl X^\mathrm{T}, \mathbfsl \Lambda \mathbfsl B|\mathbfsl Y \mathbfsl B,\mathbfsl B \mathbfsl \Theta^\mathrm{T}) \ge ms + cs - (m + c)^2.
\end{align}
Note that
\begin{align}
&H_q(\mathbfsl B \mathbfsl X^{\mathrm T}, \mathbfsl \Lambda \mathbfsl B|\mathbfsl Y \mathbfsl B,\mathbfsl B \mathbfsl \Theta^{\mathrm T}) \notag \\&= H_q(\mathbfsl B \mathbfsl X^{\mathrm T}, \mathbfsl \Lambda \mathbfsl B,\mathbfsl Y \mathbfsl B,\mathbfsl B \mathbfsl \Theta^{\mathrm T})- H_q(\mathbfsl Y \mathbfsl B,\mathbfsl B \mathbfsl \Theta^{\mathrm T})\nonumber\\
&= H_q(\mathbfsl B \mathbfsl X^{\mathrm T},\mathbfsl B \mathbfsl \Theta^\mathrm{T}) + H_q(\mathbfsl \Lambda \mathbfsl B,\mathbfsl Y \mathbfsl B|\mathbfsl B \mathbfsl X^{\mathrm T},\mathbfsl B \mathbfsl \Theta^{\mathrm T})  \notag \\ &-H_q(\mathbfsl Y \mathbfsl B,\mathbfsl B \mathbfsl \Theta^{\mathrm T}).
\end{align}
Due to the fact that $x_i$ and $\theta_i$ values are chosen from two different sets, we know that $\mathbfsl \Theta|| \mathbfsl X$ is a full-rank matrix. Therefore, we have
\begin{align}
H_q( \mathbfsl B \mathbfsl X^{\mathrm T},\mathbfsl B \mathbfsl \Theta^{\mathrm T}) = (c+m)s.
\end{align}
Similarly, the matrix $\mathbfsl \Lambda||\mathbfsl Y$ is full-rank, thus by Lemma \ref{lem:conditional_matrix_entropy} 
\begin{align}
H_q(\mathbfsl \Lambda \mathbfsl B, \mathbfsl Y \mathbfsl B|\mathbfsl B \mathbfsl X^{\mathrm T},\mathbfsl B \mathbfsl \Theta^{\mathrm T}) \ge (c+m)s  - (c+m)^2.    
\end{align}
Also, we have
\begin{align}
H_q(\mathbfsl Y \mathbfsl B,\mathbfsl B \mathbfsl \Theta^{\mathrm T})\le H_q(\mathbfsl Y \mathbfsl B) + H_q(\mathbfsl B \mathbfsl \Theta^{\mathrm T}) = ms + cs.
\end{align}
Therefore, we have
\begin{align}
H_q(\mathbfsl B \mathbfsl X^{\mathrm T}, \mathbfsl \Lambda \mathbfsl B|\mathbfsl Y \mathbfsl B, \mathbfsl B \mathbfsl \Theta^{\mathrm T}) \ge ms + cs - (m+c)^2.
\end{align}
This gives us the desired inequality $H_q(\mathbfsl A|\mathbfsl V, \mathbfsl U,\mathbfsl \Gamma,\mathbfsl \Omega)\ge d- (m + c)^2$. \\
We proved that $H_q(\mathbfsl A| (\mathbfsl A+ \mathbfsl B) \mathbfsl X^{\mathrm T},\mathbfsl Y \mathbfsl B,\mathbfsl \Gamma,\mathbfsl \Omega) \ge d - (m+c)^2$, for {\it every} choice of the evaluation points $x_1,\cdots, x_m$. Note that in general, the verifier may choose the evaluation points {\it after} observing the commitment ${(\mathbfsl \Gamma,\mathbfsl \Omega)}$. In other words, rather than assuming $(\mathbfsl X, \mathbfsl Y)$ are arbitrary constants, we must treat them as random variables  satisfying the following Markov chain
\begin{align}
(\mathbfsl A, \mathbfsl B)\longleftrightarrow (\mathbfsl \Gamma,\mathbfsl \Omega) \longleftrightarrow (\mathbfsl X, \mathbfsl Y).
\end{align}
But thanks to Lemma \ref{lem:replace_by_constant}, since $(\mathbfsl \Gamma, \mathbfsl \Omega)$ appear in the conditioning, we have
\begin{align}
    &H_q(\mathbfsl A| (\mathbfsl A+ \mathbfsl B) \mathbfsl X^{\mathrm T},\mathbfsl Y \mathbfsl B,\mathbfsl \Gamma,\mathbfsl \Omega) \ge \notag \\
    &\min_{x_1,\cdots,x_m} H_q(\mathbfsl A| (\mathbfsl A+ \mathbfsl B)\begin{bmatrix}
        1 & x_1 & \cdots & x_1^{s-1}\\
        & \cdots &\cdots & \notag \\
        1 & x_m & \cdots & x_m^{s-1}
    \end{bmatrix}^\mathrm{T}, \notag \\ &\begin{bmatrix}
        1 & x_1^s & \cdots & x_1^{s(s-1)}\\
        & \cdots &\cdots &\\ 
        1 & x_m^s & \cdots & x_m^{s(s-1)}
    \end{bmatrix} \mathbfsl B,\mathbfsl \Gamma,\mathbfsl \Omega).
\end{align}
Therefore, the analysis we provided also addresses the case where the evaluation points are chosen adaptively, after observing the commitment. 
\section{Conclusions}
\label{sec:discussion}
In this work, we have developed \name, an information-theoretic protocol for polynomial commitment and verification. Specifically, we have considered a setting with an untrusted server (prover) hosting a private polynomial $f$ of degree $d-1$ and a user (verifier) who wishes to obtain evaluations of $f$. \namespace consists of a commitment phase, in which the user learns a private commitment to the prover's polynomial, and an evaluation phase. \namespace only requires a trusted third-party in the commitment phase, provides unconditional privacy guarantees for the server and unconditional verifiability for the user irrespective of their computational power. Moreover, \namespace is doubly-efficient with an efficient prover of complexity $O(d)$ and a super-efficient verifier with complexity of $O(\sqrt{d})$. 
\section*{Acknowledgement}
The authors would like to thank Mohammad Maddah-Ali, Srivatsan Ravi and Ali Rahimi for the fruitful discussions and for revising the manuscript. This material is based upon work supported by  ARO award W911NF1810400, NSF grants CCF-1703575 and CCF-1763673, ONR Award No. N00014-16-1-2189 and research gifts from Intel, Cisco, and Qualcomm.

\bibliographystyle{IEEEtran}
\bibliography{main.bib}

\appendices
 
\section{Lemma \ref{lem:fullrank}}
\label{appendix:lemma:fullrank}
\begin{lemma}
\label{lem:fullrank}
Let $c,d,s$ be three positive integers such that $c,d\le s$. Let $\mathbfsl F$ and $\mathbfsl G$ be two full-rank $c\times s$ and $d\times s$ matrices over $\mathbb F_q$, respectively. Let $\mathbfcal F =\mathbfsl{I}_{s\times s}\otimes \mathbfsl F$, $\mathbfcal G = {\mathbfsl G}\otimes \mathbfsl{I}_{s\times s}$ and let $\mathbfcal H = {\mathbfcal F}|| {\mathbfcal G}$ be the $(c+d)s \times s^2$ matrix obtained from vertical concatenation of ${\mathbfcal F}$ and ${\mathbfcal G}$. Then, we have
\begin{align}
\rank({\mathbfcal H})\ge (c+d)s - cd.
\end{align}
\end{lemma}

\begin{proof}
Since $\mathbfsl F$ is full-rank and $c\le s$, $\mathbfsl F$ must have at least one $c\times c$ full-rank submatrix. Let $\tau \subseteq [0:s-1]$ represent the indices of the columns of one such sub-matrix and let $\mathbfsl F_{\tau}$ be the corresponding submatrix of $\mathbfsl F$. We also define  
\begin{align}
    \rho = \{cs + i + sj|j\in [0:d-1], i \in \tau\}.
\end{align}
Intuitively, $\rho$ corresponds to the rows of ${\mathbfcal G}$ which may have a non-zero element in any column from ${\tau}$. Note that $|\rho| = cd$. We will argue that by eliminating the rows indexed in $\rho$ from ${\mathbfcal H}$, we will find a full-rank matrix. Since this resulting matrix has ${(c+d)s-cd}$ rows, the claim will follow. In Equation \eqref{eqn:bigmatrix} below, we have illustrated an example with $s = 4, c = 3$ and $d=2$. We have assumed that ${\tau} = \{1,2,3\}$ and have marked the full-rank submatrix of $\mathbfsl F$ in blue. The rows indexed in $\rho$ are shown in red.

Let $\pmb \lambda$ be a vector of length $(c+d)s$ such that $\lambda_i = 0, \forall i\in\rho$ and $\pmb \lambda {\mathbfcal H} = {\bf 0}$. We will show that ${\pmb \lambda}$ must be the all-zero vector. For $j\in[0:s-1]$, let ${\mathbfcal H}_j$ be the submatrix of ${\mathbfcal H}$ obtained from columns $\{i +sj|i+cs\in\rho\}$. Since $\pmb \lambda {\mathbfcal H} = {\bf 0}$, we must have $\pmb \lambda {\mathbfcal H}_j = {\bf 0}$. But for any $i\in [0:(c+d)s-1]\backslash [jc: (j+1)c-1]$, either $\lambda_i = 0$ or the $i$'th row of ${\mathbfcal H}_j$ is an all-zero vector. It follows that $\pmb \lambda_{[jc:(j+1)c-1]} \mathbfsl F_\tau = {\bf 0}$. But $\mathbfsl F_\tau$ is full-rank, so $\pmb \lambda_{[jc:(j+1)c-1]} = {\bf 0}$. Since this holds for every $j\in [0:s-1]$, we conclude that $\pmb \lambda_{[0:cs-1]}={\bf 0}$.

It follows from $\pmb \lambda_{[0:cs-1]}={\bf 0}$ and $\pmb \lambda {\mathbfcal H} = {\bf 0}$ that $\pmb \lambda_{[cs:(c+d)s-1]}{\mathbfcal G} = {\bf 0}$. Define the $s\times d$ matrix $\pmb \theta$ such that $\theta_{i,j} = \lambda_{i + sj}$. Observe that $\pmb \lambda_{[cs:(c+d)s-1]}{\mathbfcal G}$ is a vector found by rearranging the elements of the matrix $\pmb \theta \mathbfsl G$. Since $\pmb \lambda_{[cs:(c+d)s-1]}{\mathbfcal G} = {\bf 0}$, we must have $\pmb \theta \mathbfsl G = {\bf 0}$. Since $\mathbfsl G$ is full-rank, it follows that $\pmb = {\bf 0}$ and as a result $\pmb \lambda_{[cs:(c+d)s-1]} = {\bf 0}$. We conclude that $\pmb \lambda = {\bf 0}$ which establishes the claim.
\end{proof}

\begin{strip}
{
\begin{align}
{\mathbfcal H} =
\begin{bmatrix}
    f_{0,0} & \textcolor{blue}{f_{0,1}} & \textcolor{blue}{f_{0,2}} & \textcolor{blue}{f_{0,3}} & 0 & 0 & 0& 0& 0&0&0& 0& 0&0&0& 0\\
    f_{1,0} & \textcolor{blue}{f_{1,1}} & \textcolor{blue}{f_{1,2}} & \textcolor{blue}{f_{1,3}}& 0 & 0 & 0& 0& 0&0&0& 0& 0&0&0& 0\\
    f_{2,0} & \textcolor{blue}{f_{2,1}} & \textcolor{blue}{f_{2,2}} & \textcolor{blue}{f_{2,3}}& 0 & 0 & 0& 0& 0&0&0& 0& 0&0&0& 0\\
    0&0&0& 0& f_{0,0} & \textcolor{blue}{f_{0,1}} & \textcolor{blue}{f_{0,2}} & \textcolor{blue}{f_{0,3}} & 0 & 0 & 0& 0& 0&0&0& 0\\
    0&0&0& 0 &f_{1,0} & \textcolor{blue}{f_{1,1}} & \textcolor{blue}{f_{1,2}} & \textcolor{blue}{f_{1,3}}& 0 & 0 & 0& 0& 0&0&0& 0\\
    0&0&0& 0 &f_{2,0} & \textcolor{blue}{f_{2,1}} & \textcolor{blue}{f_{2,2}} & \textcolor{blue}{f_{2,3}}& 0 & 0 & 0& 0& 0&0&0& 0\\
    0&0&0& 0 &0&0&0& 0& f_{0,0} & \textcolor{blue}{f_{0,1}} & \textcolor{blue}{f_{0,2}} & \textcolor{blue}{f_{0,3}} & 0 & 0 & 0& 0\\
    0&0&0& 0 &0&0&0& 0 &f_{1,0} &\textcolor{blue}{f_{1,1}} & \textcolor{blue}{f_{1,2}} & \textcolor{blue}{f_{1,3}}& 0 & 0 & 0& 0\\
    0&0&0& 0 &0&0&0& 0 &f_{2,0} & \textcolor{blue}{f_{2,1}} & \textcolor{blue}{f_{2,2}} & \textcolor{blue}{f_{2,3}}& 0 & 0 & 0& 0\\
     0 & 0 & 0& 0& 0&0&0& 0& 0&0&0& 0 & f_{0,0} & \textcolor{blue}{f_{0,1}} & \textcolor{blue}{f_{0,2}} & \textcolor{blue}{f_{0,3}}\\
     0 & 0 & 0& 0& 0&0&0& 0& 0&0&0& 0 & f_{1,0} & \textcolor{blue}{f_{1,1}} & \textcolor{blue}{f_{1,2}} & \textcolor{blue}{f_{1,3}}\\
     0 & 0 & 0& 0& 0&0&0& 0& 0&0&0& 0 & f_{2,0} & \textcolor{blue}{f_{2,1}} & \textcolor{blue}{f_{2,2}} & \textcolor{blue}{f_{2,3}}\\
     g_{0,0} & 0 & 0 &0 & g_{0,1} & 0 & 0 & 0& g_{0,2} & 0 & 0 & 0 & g_{0,3} & 0 & 0 & 0 \\
    \textcolor{red}{0} & \textcolor{red}{g_{0,0}} & \textcolor{red}{0} & \textcolor{red}{0} &\textcolor{red}{0} & \textcolor{red}{g_{0,1}} & \textcolor{red}{0} & \textcolor{red}{0} & \textcolor{red}{0}& \textcolor{red}{g_{0,2}} & \textcolor{red}{0} & \textcolor{red}{0} & \textcolor{red}{0} & \textcolor{red}{g_{0,3}} & \textcolor{red}{0}& \textcolor{red}{0}\\
     \textcolor{red}{0} & \textcolor{red}{0} & \textcolor{red}{g_{0,0}} & \textcolor{red}{0} & \textcolor{red}{0} &\textcolor{red}{0} & \textcolor{red}{g_{0,1}} & \textcolor{red}{0} & \textcolor{red}{0} & \textcolor{red}{0}& \textcolor{red}{g_{0,2}} & \textcolor{red}{0} & \textcolor{red}{0} & \textcolor{red}{0} & \textcolor{red}{g_{0,3}} & \textcolor{red}{0}\\
     \textcolor{red}{0} & \textcolor{red}{0} & \textcolor{red}{0} & \textcolor{red}{g_{0,0}} & \textcolor{red}{0} & \textcolor{red}{0} &\textcolor{red}{0} & \textcolor{red}{g_{0,1}} & \textcolor{red}{0} & \textcolor{red}{0} & \textcolor{red}{0}& \textcolor{red}{g_{0,2}} & \textcolor{red}{0} & \textcolor{red}{0} & \textcolor{red}{0} & \textcolor{red}{g_{0,3}}\\
     g_{1,0} & 0 & 0 &0 & g_{1,1} & 0 & 0 & 0& g_{1,2} & 0 & 0 & 0 & g_{1,3} & 0 & 0 & 0 \\
    \textcolor{red}{0} & \textcolor{red}{g_{1,0}} & \textcolor{red}{0} & \textcolor{red}{0} &\textcolor{red}{0} & \textcolor{red}{g_{1,1}} & \textcolor{red}{0} & \textcolor{red}{0} & \textcolor{red}{0}& \textcolor{red}{g_{1,2}} & \textcolor{red}{0} & \textcolor{red}{0} & \textcolor{red}{0} & \textcolor{red}{g_{1,3}} & \textcolor{red}{0}& \textcolor{red}{0}\\
     \textcolor{red}{0} & \textcolor{red}{0} & \textcolor{red}{g_{1,0}} & \textcolor{red}{0} & \textcolor{red}{0} &\textcolor{red}{0} & \textcolor{red}{g_{1,1}} & \textcolor{red}{0} & \textcolor{red}{0} & \textcolor{red}{0}& \textcolor{red}{g_{1,2}} & \textcolor{red}{0} & \textcolor{red}{0} & \textcolor{red}{0} & \textcolor{red}{g_{1,3}} & \textcolor{red}{0}\\
     \textcolor{red}{0} & \textcolor{red}{0} & \textcolor{red}{0} & \textcolor{red}{g_{1,0}} & \textcolor{red}{0} & \textcolor{red}{0} &\textcolor{red}{0} & \textcolor{red}{g_{1,1}} & \textcolor{red}{0} & \textcolor{red}{0} & \textcolor{red}{0}& \textcolor{red}{g_{1,2}} & \textcolor{red}{0} & \textcolor{red}{0} & \textcolor{red}{0} & \textcolor{red}{b_{1,3}}
\end{bmatrix}.
\label{eqn:bigmatrix}
\end{align}
}
\end{strip}

\section{Adaptivity of \name \\ (Proof of Theorem \ref{theorem:Adaptivity InfoCommit})}
\label{appendix: Proof of Adaptivity thm}

We now provide the proof of Theorem \ref{theorem:Adaptivity InfoCommit}. 

\begin{proof}

Suppose that an adversarial server responds with $\mathcal V_\ell=\mathbfsl v_\ell+\mathcal W_\ell$ and $\mathcal U_\ell=\mathbfsl u_\ell+\mathcal Z_\ell$, where $(\mathcal W_l, \mathcal Z_l) \neq (\mathbfsl 0, \mathbfsl 0)$. That is, the server returns at least one incorrect computation. Since the server is informed whether the computation outcome is accepted or not after each evaluation, then we have the following Markov property 
\begin{align}
\label{eqn:markov-property}
(\mathbfsl \Lambda, \mathbfsl \Theta) \leftrightarrow	 (V_1, V_2, \cdots, V_{\ell-1}) \leftrightarrow	(\mathcal W_\ell, \mathcal Z_\ell),
\end{align}
where $V_\ell=\mathbf 1(\mathbfsl \Lambda \mathcal W_\ell=\mathbfsl 0, \mathcal Z_\ell \mathbfsl \Theta^{\mathrm T}=\mathbfsl 0)$. The probability that the server fails all verifications in all $m$ rounds can be expressed as follows 
\begin{align}
P_m &=\Pr \left( V_m = 0, V_{m-1}=0 \cdots, V_1=0 \right) \notag \\
&=\Pr \left( V_m = 0| V_{m-1}=0 \cdots, V_1=0 \right)  P_{m-1}.
\end{align}
The first term can be lower-bounded as follows
\begin{align}
    &T_m \coloneqq \Pr \left( V_m = 0| V_{m-1}=0 \cdots, V_1=0 \right) \notag \\
        &=\sum_{ (\mathbfsl w_m, \mathbfsl z_m) \neq (\mathbfsl 0, \mathbfsl 0) } \Pr \left( (\mathbfsl \Lambda \mathcal W_m \neq \mathbfsl 0) \cup  (\mathcal Z_m \mathbfsl \Theta^{\mathrm T} \neq \mathbfsl 0)| \right. \notag \\ &~~~~~~~ \left. (\mathcal W_m, \mathcal Z_m)=(\mathbfsl w_m, \mathbfsl z_m), V_{m-1}=0, \cdots, V_{1}=0 \right)  \notag \\
        & ~~~~~~~~~~~~~~~~\Pr((\mathcal W_m, \mathcal Z_m)=(\mathbfsl w_m, \mathbfsl z_m)|V_{m-1}=0, \cdots, V_{1}=0) \notag \\
        &=\sum_{ (\mathbfsl w_m, \mathbfsl z_m) \neq (\mathbfsl 0, \mathbfsl 0) } \Pr \left( (\mathbfsl \Lambda \mathbfsl w_m \neq \mathbfsl 0) \cup  (\mathbfsl z_m \mathbfsl \Theta^{\mathrm T} \neq \mathbfsl 0)| \right. \notag \\ & \left. ~~~~~~~~~~~~~~~~~~~~~~~~~~ V_{m-1}=0, \cdots, V_{1}=0 \right)  \notag \\
        & ~~~~~~~~~~~~~~~~~~~~\Pr((\mathcal W_m, \mathcal Z_m)=(\mathbfsl w_m, \mathbfsl z_m)|V_{m-1}=0, \cdots, V_{1}=0),
\end{align}
where the last equality follows from the Markov property in (\ref{eqn:markov-property}). The first term inside this summation can be bounded as  
\begin{align}
R_m &\coloneqq \Pr \left( (\mathbfsl \Lambda \mathbfsl w_m \neq \mathbfsl 0) \cup  (\mathbfsl z_m \mathbfsl \Theta^{\mathrm T} \neq \mathbfsl 0)| V_{m-1}=0, \cdots, V_{1}=0 \right) \notag \\
&= 1-\Pr \left( \mathbfsl \Lambda \mathbfsl w_m = \mathbfsl 0, \mathbfsl z_m \mathbfsl \Theta^{\mathrm T} = \mathbfsl 0| V_{m-1}=0, \cdots, V_{1}=0 \right)\notag \\
&=1-\frac{\Pr \left( \mathbfsl \Lambda \mathbfsl w_m = \mathbfsl 0, \mathbfsl z_m \mathbfsl \Theta^{\mathrm T} = \mathbfsl 0, V_{m-1}=0, \cdots, V_{1}=0 \right)}{P_{m-1}} \notag \\
& \geq 1-\frac{\Pr \left( \mathbfsl \Lambda \mathbfsl w_m = \mathbfsl 0, \mathbfsl z_m \mathbfsl \Theta^{\mathrm T} = \mathbfsl 0 \right)}{P_{m-1}} \notag \\
&\geq 1-\frac{1}{P_{m-1}}\left(\frac{2}{r^{c}}+\frac{1}{r^{2c}} \right).
\end{align}
Hence, we have
\begin{align}
    T_m \geq  1-\frac{1}{P_{m-1}} \left( \frac{2}{r^{c}}+\frac{1}{r^{2c}} \right)
\end{align}
and 
\begin{align}
    P_{m} &= T_m P_{m-1}\notag \\ 
          &\geq P_{m-1}-\left( \frac{2}{r^{c}}+\frac{1}{r^{2c}} \right) \notag \\
          & \geq P_1 -(m-1)\left( \frac{2}{r^{c}}+\frac{1}{r^{2c}} \right) \notag \\
          & \geq 1-m\left( \frac{2}{r^{c}}+\frac{1}{r^{2c}} \right),
\end{align}
where the last inequality follows since $P_1 \geq 1-\left( \frac{2}{r^{c}}+\frac{1}{r^{2c}} \right)$. Hence, we have    
 \begin{align}
        \Pr(&\exists \ell \in[m] \ \text{s.t.} \ \text{\bf Verify}(x_\ell,\hat{val}_\ell,\text{VK}, K_v) = 1,\hat{val}_\ell \neq {val}_\ell) \notag \\ &\le m \left( \frac{2}{r^{c}}+\frac{1}{r^{2c}} \right).
\end{align}

\end{proof}

\end{document}